\documentclass[aps,pra,twocolumn,superscriptaddress,nofootinbib,longbibliography]{revtex4-2}

\usepackage{amsmath,amssymb,amsthm}
\usepackage{graphicx}
\usepackage[hidelinks]{hyperref}

\newtheorem{theorem}{Theorem}
\newtheorem{corollary}{Corollary}
\newtheorem{definition}{Definition}

\newcommand{\F}{\mathbb{F}}
\newcommand{\Z}{\mathbb{Z}}
\newcommand{\Sp}{\mathrm{Sp}}
\newcommand{\GL}{\mathrm{GL}}

\begin{document}

\title{Computational Complexity of Clifford Template Compilation:\texorpdfstring{\\}{ }%
Are Quantum Computers Useful for Compiling Quantum Circuits?}

\author{Keisuke Fujii}
\affiliation{Graduate School of Informatics, Kyoto University, Sakyo-ku, Kyoto, 606-8501, Japan}
\affiliation{Graduate School of Engineering Science, Osaka University, Toyonaka, Osaka 560-8531, Japan}
\affiliation{Center for Quantum Information and Quantum Biology, Osaka University, Toyonaka, Osaka 560-8531, Japan}
\affiliation{RIKEN Center for Quantum Computing (RQC), Hirosawa 2-1, Wako, Saitama 351-0198, Japan}

\date{September 27, 2026}

\begin{abstract}
A Clifford template is a finite ordered family of repeatable Clifford operations, and an instantiation specifies how many times each operation is applied.  The Clifford template compilation problem asks how to choose these repetition numbers so that the template realizes a target transformation of Pauli operators.  This problem arises, for example, when searching for logical operations in quantum error correction using only Clifford operations permitted by physical or fault-tolerance constraints.  Although forward Clifford dynamics is efficiently classically simulable, this inverse problem has sharp complexity transitions.  For commuting templates with unrestricted integer exponents, feasibility lies in $\mathrm{NP}\cap\mathrm{BQP}$ and a constructive quantum algorithm returns a particular solution together with the full exponent-relation lattice; already at $k=1$, recovering the repetition number contains finite-field discrete logarithm over $\F_{2^r}^{\times}$.  In general, restricting every exponent to $\{0,1\}$ removes the Abelian-group closure and makes feasibility NP-complete for variable $k$, even for exactly commuting CNOT-only operations and X-type Paulis.  For commuting self-inverse Clifford actions, both binary feasibility and recovery of one solution are classically polynomial-time solvable, but imposing a bound on the total repetition count is NP-complete, even for CNOT-only operations.  These results reveal a rich complexity landscape within Clifford template compilation, spanning classically tractable cases, problems admitting quantum polynomial-time algorithms, and NP-complete variants.
\end{abstract}

\maketitle

\section{Introduction}

\begin{figure*}[t]
\centering
\includegraphics[trim=130.22bp 3.25bp 135.88bp 7.88bp,clip,angle=-90,width=\textwidth]{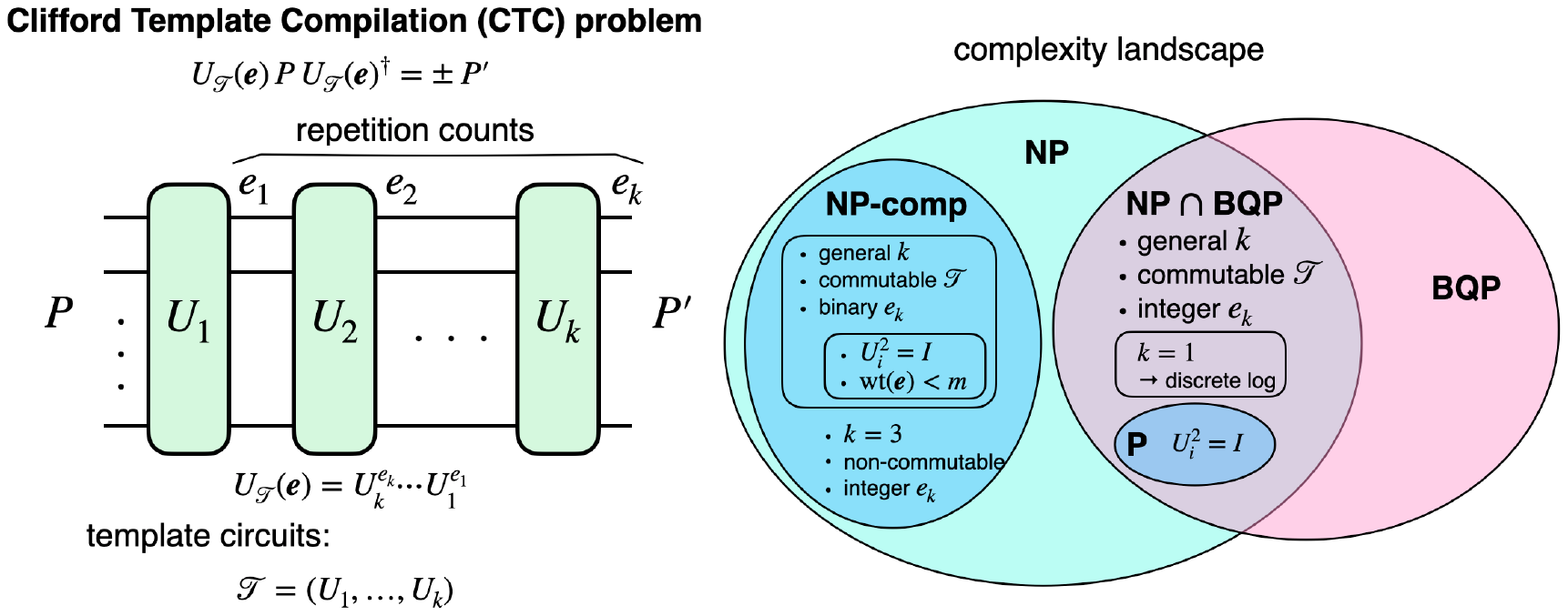}
\caption{Clifford template compilation and its complexity landscape.  Left: repetition counts instantiate an ordered Clifford template to map an input Pauli operator $P$ to a target $P'$, up to phase.  Right: commuting integer-exponent decision CTC lies in $\mathrm{NP}\cap\mathrm{BQP}$, whereas its binary-exponent variant is NP-complete for variable $k$.  For commuting self-inverse actions, feasibility is in $\mathrm P$, but imposing a bound on the number of template-operation applications is NP-complete.  Noncommuting integer-exponent decision CTC is NP-complete already at $k=3$.  The $k=1$ discrete-logarithm annotation refers to exponent recovery (search).  Commutativity is defined for the induced phase-free Pauli actions; the displayed condition $U_i^2=I$ is a sufficient self-inverse restriction.}
\label{fig:fig01}
\end{figure*}

Determining whether a desired quantum operation can be synthesized from a prescribed set of elementary operations is a basic problem of quantum gate synthesis and compilation.  For a generic $n$-qubit target, however, the worst-case circuit size grows exponentially with $n$ \cite{shende2006}.  Important applications nevertheless contain structured Clifford-compilation subproblems.  Clifford operations are central to quantum error correction with the stabilizer codes \cite{gottesman1997}, where they implement logical transformations, including addressable operations on selected logical qubits \cite{breuckmann2024,tansuwannont2026}.  Furthermore, in product-formula Hamiltonian simulation, many-body Pauli rotations are commonly implemented using Clifford basis changes around single-qubit rotations \cite{mukhopadhyay2023}.

In practice, the available Clifford operations are often restricted by hardware or fault-tolerance constraints.  Superconducting processors typically provide entangling gates between neighboring qubits, neutral-atom arrays allow structured shuttling of atoms, and fault-tolerant constructions favor protected logical operations such as transversal, fold-transversal, or code-automorphism gates \cite{arute2019,bluvstein2022,breuckmann2024,tansuwannont2026}.  These settings naturally supply finite families of structured Clifford operations that can be reused.  We call such a prescribed finite family a \emph{Clifford template}.  A natural inverse problem is to choose repetition counts that realize a desired physical or logical Clifford action.

At first sight this problem may appear classically easy.  Clifford circuits admit efficient tableau simulation by the Gottesman-Knill theorem, and the action of an $n$-qubit Clifford on Pauli operators is a binary symplectic linear map \cite{gottesman1997,aaronson2004}.  Thus, for any fixed repetition pattern one can propagate Paulis and verify the answer efficiently.  The subtlety has two sources.  First, a repetition count specified using polynomially many bits may have exponentially large numerical value, so that it can represent exponentially many applications of a single operation.  Second, when the number of available operations scales with the input, the number of possible repetition patterns can itself be exponential; even choosing only whether to apply each operation once or not at all gives exponentially many patterns.  The compilation problem asks us to invert these compressed families of classically simulable dynamics.

In this work, we study the Clifford Template Compilation (CTC) problem in a restricted but fundamental setting.  Given an ordered template $\mathcal{T}=(U_1,\ldots,U_k)$, its instantiation is specified by repetition counts $e_1,\ldots,e_k$ and acts as $U_{\mathcal{T}}(\boldsymbol e)=U_k^{e_k}\cdots U_1^{e_1}$.  The desired Clifford behavior is specified only by the image of a single Pauli operator $P$,
\begin{equation}
  U_{\mathcal{T}}(\boldsymbol e) P U_{\mathcal{T}}(\boldsymbol e)^\dagger=P',
\end{equation}
We refer to this restricted task as \emph{Pauli reachability}.  It specifies only the image of $P$, rather than the full action of $U_{\mathcal{T}}(\boldsymbol e)$.  We analyze both general ordered templates and the important restriction in which the induced phase-free Pauli actions commute.  We also distinguish unrestricted integer repetition counts from the binary-exponent variant in which every operation may be selected at most once.  As a subproblem of full Clifford template compilation, Pauli reachability yields lower bounds that apply to more expressive template-compilation tasks.  Figure~\ref{fig:fig01} gives an overview of the task and the complexity regimes studied here.

The problem exhibits a rich complexity landscape as we vary three simple features: whether the template operations commute, whether the exponents are binary or unrestricted integers, and whether the task is to decide existence or to find an explicit instantiation.  Across these closely related regimes, decision CTC ranges from NP-complete cases to a commuting integer-exponent case in $\mathrm{NP}\cap\mathrm{BQP}$, while the unrestricted integer search problem already contains finite-field discrete logarithm for a one-operation template.  Thus small changes in the formulation produce sharp complexity transitions within the same Clifford--Pauli reachability framework.

Our results are summarized as follows.
\begin{itemize}
\item If the template size $k$ is part of the input, binary-exponent commuting decision CTC is NP-complete even for CNOT-only operations and X-type Paulis.  Its search version is polynomial-time Turing reducible to the decision version using at most $k$ decision queries.
\item For commuting templates with unrestricted integer exponents, decision CTC is in $\mathrm{NP}\cap\mathrm{BQP}$ for arbitrary template size $k$.  Standard quantum algorithms for finite Abelian groups also solve search CTC directly in polynomial time and produce the full relation lattice \cite{shor1997,cheung2001,ivanyos2003,childs2014}.
\item For commuting self-inverse Clifford actions, binary and integer reachability coincide, and one binary solution can be found in deterministic classical polynomial time.  In contrast, deciding whether a solution uses at most $m$ template-operation applications is NP-complete, even for exactly commuting CNOT-only self-inverse operations (Sec.~\ref{sec:involutions}).
\item For a one-operation template ($k=1$), search CTC over qubits contains discrete logarithm in $\F_{2^r}^{\times}$.  The same construction for prime-dimensional qudits gives DLP in $\F_{p^r}^{\times}$, including $\F_p^\times$.
\item Without the commutativity restriction, integer-exponent decision CTC is NP-complete already for $k=3$.  The hardness holds for wire-permutation Clifford operations, hence for SWAP-only circuits, and follows by embedding membership in a product of three cyclic permutation groups.  The general noncommuting case $k=2$ remains open.
\end{itemize}

The main message is, of course, not that an NP-complete language is efficiently solved by a quantum computer.  Rather, CTC gives a single natural problem family in which Shor-type structured witness recovery and NP-complete Boolean selection occur in neighboring regimes.  The two hardness embeddings use different template-size and output regimes---integer search at $k=1$ for DLP, and binary decision at variable $k$ for NP-completeness---but they are expressed by the same Clifford--Pauli reachability relation rather than by unrelated encodings.  The distinctive point is therefore not either ingredient in isolation, but their realization as adjacent slices of one CNOT-only, phase-free Clifford framework.  In the binary-linear subfamily, moreover, the unitaries appearing in phase estimation are CNOT-only.  Coherently controlling them introduces Toffoli gates, but with substantially less non-Clifford overhead than modular-arithmetic tasks such as integer factoring, making this subfamily a promising candidate for an efficiently verifiable workload on early fault-tolerant quantum computers.

\section{Pauli and Clifford representation}

We represent an $n$-qubit Pauli operator up to phase by a vector
\begin{equation}
  a=(x,z)\in \F_2^{2n},
\end{equation}
where $x,z\in \F_2^n$ specify the X and Z supports.  An $n$-qubit Clifford unitary $U$ acts by conjugation as
\begin{equation}
  U P(a) U^\dagger = \pm P(Fa),
\end{equation}
where $F\in \Sp(2n,\F_2)$ preserves the standard symplectic form.  In this work phases are ignored.  For the NP-completeness construction below all gates are CNOTs and all Paulis are X-type, so the construction can also be read as a statement about reversible binary linear circuits embedded into Clifford circuits.

\begin{definition}[Clifford template]
A Clifford template is an ordered finite tuple
\begin{equation}
  \mathcal{T}=(U_1,\ldots,U_k)
\end{equation}
of Clifford operations designated for repeated use.  An exponent vector $\boldsymbol e=(e_1,\ldots,e_k)$ instantiates the template as
\begin{equation}
  U_{\mathcal{T}}(\boldsymbol e)=U_k^{e_k}\cdots U_1^{e_1}.
\end{equation}
We call $k$ the template size.  Throughout this work, ``commuting'' means commuting at the level of phase-free Pauli actions: if $F_i$ is the symplectic matrix induced by $U_i$, then $F_iF_j=F_jF_i$ for all $i,j$.  Exact commutativity of the Clifford unitaries implies this condition, but the converse need not hold; the unitary commutator may be a nontrivial Pauli operator up to global phase.
\end{definition}

\begin{definition}[Decision Clifford template compilation, $\mathrm{CTC}_{\mathrm{dec}}$]
An instance consists of a Clifford template $\mathcal{T}=(U_1,\ldots,U_k)$ and Pauli operators $P$ and $P'$.  The question is whether there exist nonnegative integers $e_1,\ldots,e_k$ such that
\begin{equation}
  U_{\mathcal{T}}(\boldsymbol e)P U_{\mathcal{T}}(\boldsymbol e)^\dagger=\pm P'.
  \label{eq:ctc}
\end{equation}
Computationally, the instance is represented by the induced matrices $F_1,\ldots,F_k\in\Sp(2n,\F_2)$ and the phase-free Pauli vectors $a,b\in\F_2^{2n}$.  The condition is equivalently
\begin{equation}
  F_k^{e_k}\cdots F_1^{e_1} a=b.
\end{equation}
We call this decision problem \emph{decision Clifford Template Compilation}, denoted $\mathrm{CTC}_{\mathrm{dec}}$, and use the symplectic notion of commutativity defined above.  Since every $F_i$ belongs to a finite group, allowing $e_i\in\Z$ gives the same decision problem: each exponent may be reduced modulo $\operatorname{ord}(F_i)$, the smallest positive integer $t$ such that $F_i^t=I$.  The \emph{binary-exponent variant} restricts every exponent to $e_i\in\{0,1\}$.
\end{definition}

\begin{definition}[Search Clifford template compilation, $\mathrm{CTC}_{\mathrm{search}}$]
For the same input, search CTC asks, on every yes-instance, for one exponent vector $\boldsymbol e$ satisfying Eq.~\eqref{eq:ctc}, with each coordinate chosen in the bounded range
\begin{equation}
  0\leq e_i<|\Sp(2n,\F_2)|.
  \label{eq:bounded-search-range}
\end{equation}
Such a representative always exists by periodicity; on a no-instance the set of valid outputs is empty.  We use \emph{search CTC} and \emph{constructive CTC} synonymously when only one witness is requested.  Producing a particular solution $\boldsymbol e^{(0)}$ together with a basis of all exponent relations is a stronger relation-lattice output task.  Below, an unqualified language-complexity statement about CTC refers to $\mathrm{CTC}_{\mathrm{dec}}$; search claims are stated explicitly.
\end{definition}

Decision CTC is in NP, and so are its commuting and binary-exponent variants.  This follows directly from the efficient classical simulation of Clifford circuits together with their finite periodicity.  If a solution exists, each exponent can be reduced modulo $\operatorname{ord}(F_i)\le |\Sp(2n,\F_2)|$, and hence has $O(n^2)$ bits.  Given such an exponent vector, one evaluates the ordered product by repeated squaring and checks Eq.~\eqref{eq:ctc} in polynomial time.  The corresponding search relation is therefore polynomially balanced and polynomial-time verifiable, and hence lies in FNP.  We use FBQP in the relation-search sense: on a yes-instance, a bounded-error quantum polynomial-time algorithm outputs a valid witness.

\section{NP-completeness of binary commuting CTC}

We first show that the binary-exponent variant of commuting decision CTC becomes NP-complete when the template size $k$ is allowed to scale.  The reduction is from Positive 1-in-3 SAT.  An instance is a conjunction
\begin{equation}
  \Phi=C_1\wedge\cdots\wedge C_m
\end{equation}
of constraints called \emph{clauses}.  Each clause is a triple of Boolean variables,
\begin{equation}
  C_c=(x_{i_c},x_{j_c},x_{\ell_c}),
\end{equation}
and ``positive'' means that none of the variables is negated.  An assignment satisfies $C_c$ when exactly one of its three variables is true, or equivalently when
\begin{equation}
  e_{i_c}+e_{j_c}+e_{\ell_c}=1,
  \qquad e_i\in\{0,1\}.
\end{equation}
It satisfies $\Phi$ when this condition holds for every clause simultaneously.  For example, the clause $(x_1,x_2,x_3)$ is satisfied by $(e_1,e_2,e_3)=(1,0,0)$, $(0,1,0)$, or $(0,0,1)$, but not by $(0,0,0)$ or by an assignment with two or three true variables.  Variables may appear in several clauses, coupling these local exactly-one constraints.  Positive 1-in-3 SAT is NP-complete \cite{schaefer1978}.

\subsection{A three-qubit clause gadget}

For each clause $C_c$ introduce three qubits $(c,1),(c,2),(c,3)$ and define
\begin{equation}
  G_c=\mathrm{CNOT}_{(c,3)\rightarrow(c,2)}
      \mathrm{CNOT}_{(c,2)\rightarrow(c,1)} .
\end{equation}
The corresponding circuit is shown in Fig.~\ref{fig:clause-gadget}.
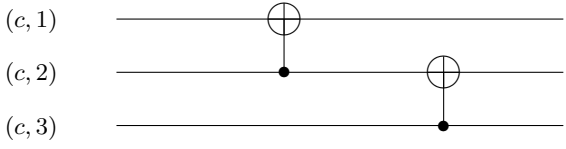
\begin{figure}[t]
\centering
\setlength{\unitlength}{1pt}
\begin{picture}(220,64)
  \put(0,49){$(c,1)$}
  \put(0,29){$(c,2)$}
  \put(0,9){$(c,3)$}
  \put(42,52){\line(1,0){168}}
  \put(42,32){\line(1,0){168}}
  \put(42,12){\line(1,0){168}}
  \put(105,32){\circle*{4}}
  \put(105,32){\line(0,1){20}}
  \put(105,52){\circle{12}}
  \put(99,52){\line(1,0){12}}
  \put(105,46){\line(0,1){12}}
  \put(165,12){\circle*{4}}
  \put(165,12){\line(0,1){20}}
  \put(165,32){\circle{12}}
  \put(159,32){\line(1,0){12}}
  \put(165,26){\line(0,1){12}}
\end{picture}
\caption{Three-qubit clause gadget implementing $G_c$.  Circuit time runs from left to right.}
\label{fig:clause-gadget}
\end{figure}
On X supports in the ordered basis $((c,1),(c,2),(c,3))$, this acts as
\begin{equation}
  J=
  \begin{pmatrix}
  1&1&0\\
  0&1&1\\
  0&0&1
  \end{pmatrix}
  \in \GL(3,\F_2).
\end{equation}
We take $(0,0,1)^T$, corresponding to the Pauli operator $I_{(c,1)}I_{(c,2)}X_{(c,3)}$, as the initial X-support vector.  Its orbit under $J$ is
\begin{align}
J^0(0,0,1)^T&=(0,0,1)^T,\\
J^1(0,0,1)^T&=(0,1,1)^T,\\
J^2(0,0,1)^T&=(1,0,1)^T,\\
J^3(0,0,1)^T&=(1,1,1)^T,
\end{align}
and $J^4=I$.  Hence the four possible numbers of selected variables in a clause, $0,1,2,3$, are distinguished by the resulting X-type Pauli.  As explained in detail in the next subsection, each clause contains three binary variables.  It therefore suffices to count up to three to determine how many of them are assigned $1$.  The desired exactly-one value is represented by $(0,1,1)^T$, corresponding to the Pauli operator $I_{(c,1)}X_{(c,2)}X_{(c,3)}$.

\subsection{Variable-size template}

Let the Positive 1-in-3 SAT instance have variables $x_1,\ldots,x_k$ and clauses $C_1,\ldots,C_m$.  For each variable $x_i$, define a Clifford operation
\begin{equation}
  U_i=\prod_{\{c\mid x_i\in C_c\}} G_c .
\end{equation}
The resulting CTC template is $\mathcal{T}=(U_1,\ldots,U_k)$.  Identifying a Boolean assignment with $\boldsymbol x=(x_1,\ldots,x_k)\in\{0,1\}^k$, we instantiate the template as $U_{\mathcal{T}}(\boldsymbol x)=U_k^{x_k}\cdots U_1^{x_1}$.  Thus the SAT variable $x_i$ directly specifies whether $U_i$ is selected.

The product is over disjoint clause blocks except where the same clause is involved, and on a shared block every $U_i$ applies either $G_c$ or $I$.  Therefore all $U_i$ commute.  In fact, these circuits commute exactly as unitaries, which is stronger than the phase-free commutativity required by our definition.  Moreover every $U_i$ is CNOT-only and has order dividing four.

Define the initial and target Paulis
\begin{equation}
  P_a=\prod_{c=1}^{m} X_{(c,3)},\qquad
  P_b=\prod_{c=1}^{m} X_{(c,2)}X_{(c,3)} .
\end{equation}
In clause $C_c$, the total exponent applied to $G_c$ is
\begin{equation}
  s_c=\sum_{\{i\mid x_i\in C_c\}} x_i\in\{0,1,2,3\}.
\end{equation}
By the orbit table above, the clause block maps $X_{(c,3)}$ to $X_{(c,2)}X_{(c,3)}$ if and only if $s_c=1$.  Thus
\begin{equation}
  U_{\mathcal{T}}(\boldsymbol x)P_a U_{\mathcal{T}}(\boldsymbol x)^\dagger=P_b
\end{equation}
if and only if the SAT instance has an exactly-one satisfying assignment.

\begin{theorem}
Binary-exponent commuting decision CTC is NP-complete when the template size $k$ is part of the input.  NP-hardness holds even when all template operations are CNOT-only Clifford circuits of order dividing four and both input and target Paulis are X-type.
\end{theorem}

\begin{proof}
Membership in NP was noted above.  The construction maps a Positive 1-in-3 SAT instance to a binary-exponent commuting CTC instance in polynomial time.  The preceding argument shows that satisfying assignments are in one-to-one correspondence with Boolean vectors $\boldsymbol x\in\{0,1\}^k$ satisfying the Pauli-orbit equation.  Hence binary-exponent commuting CTC is NP-hard.
\end{proof}

The binary restriction makes witness recovery directly self-reducible to decision.  Set $a_1=a$ and, at step $i$, maintain the invariant
\begin{equation}
  F_k^{e_k}\cdots F_i^{e_i}a_i=b.
  \label{eq:binary-self-reduction-invariant}
\end{equation}
Starting from a yes-instance, process the generators in chronological order $i=1,\ldots,k$.  For $F_i$, query the decision oracle on the remaining template $(F_{i+1},\ldots,F_k)$ with initial vector $a_i$ and target $b$, which tests whether a solution with $e_i=0$ remains.  If the answer is yes, set $e_i=0$ and $a_{i+1}=a_i$.  If it is no, the invariant and the promise that the current instance is a yes-instance force $e_i=1$; set $a_{i+1}=F_i a_i$ and continue.  For an empty remaining template, the query simply tests $a_i=b$.  After at most $k$ queries this recovers a complete witness.  Thus there is a polynomial-time Turing reduction from binary search CTC to binary decision CTC using at most $k$ decision queries.
This argument respects the given order and therefore applies to both commuting and noncommuting binary templates.

\subsection{A complete nine-qubit example}

Consider the Positive 1-in-3 SAT formula
\begin{equation}
  \Phi=C_1\wedge C_2\wedge C_3
\end{equation}
with
\begin{align}
  C_1&=(x_1,x_2,x_3), &
  C_2&=(x_1,x_4,x_5),\nonumber\\
  C_3&=(x_2,x_4,x_6).
\end{align}
Assign three qubits to each clause: qubits $(1,2,3)$ to $C_1$, $(4,5,6)$ to $C_2$, and $(7,8,9)$ to $C_3$.  The three clause operations are
\begin{align}
  G_1&=\mathrm{CNOT}_{3\rightarrow2}\mathrm{CNOT}_{2\rightarrow1},\nonumber\\
  G_2&=\mathrm{CNOT}_{6\rightarrow5}\mathrm{CNOT}_{5\rightarrow4},\nonumber\\
  G_3&=\mathrm{CNOT}_{9\rightarrow8}\mathrm{CNOT}_{8\rightarrow7}.
\end{align}
Following the occurrence pattern of the six variables, define the template operations
\begin{align}
  U_1&=G_1G_2, & U_2&=G_1G_3, & U_3&=G_1,\nonumber\\
  U_4&=G_2G_3, & U_5&=G_2,    & U_6&=G_3.
\end{align}
They commute because the $G_c$ act on disjoint clause blocks.  The initial and target Paulis are
\begin{equation}
  P_a=X_3X_6X_9,\qquad
  P_b=(X_2X_3)(X_5X_6)(X_8X_9).
\end{equation}
The assignment
\begin{equation}
  (x_1,x_2,x_3,x_4,x_5,x_6)=(1,0,0,0,0,1)
\end{equation}
sets only $x_1$ and $x_6$ to true.  Thus $x_1$ is the unique true variable in both $C_1$ and $C_2$, while $x_6$ is the unique true variable in $C_3$.  Correspondingly,
\begin{equation}
  U_6U_1=G_3G_1G_2,\qquad
  U_6U_1P_aU_1^\dagger U_6^\dagger=P_b.
\end{equation}

The same instance can be written explicitly in the binary symplectic representation.  Let $I_3$ denote the $3\times3$ identity and let $J$ be the clause matrix defined above.  On the nine-dimensional X support, the six template operations act as
\begin{align}
  A_1&=J\oplus J\oplus I_3,
  &A_2&=J\oplus I_3\oplus J,\nonumber\\
  A_3&=J\oplus I_3\oplus I_3,
  &A_4&=I_3\oplus J\oplus J,\nonumber\\
  A_5&=I_3\oplus J\oplus I_3,
  &A_6&=I_3\oplus I_3\oplus J.
\end{align}
The corresponding symplectic matrices and Pauli vectors are
\begin{equation}
  F_i=
  \begin{pmatrix}
    A_i&0\\
    0&A_i^{-T}
  \end{pmatrix}
  \in\Sp(18,\F_2),
\end{equation}
\begin{align}
  a&=(0,0,1,\,0,0,1,\,0,0,1;\,\mathbf{0}_9)^T,\nonumber\\
  b&=(0,1,1,\,0,1,1,\,0,1,1;\,\mathbf{0}_9)^T.
\end{align}
Here the semicolon separates the X and Z supports.
For this Boolean assignment, the Pauli vectors above obey $F_6F_1a=b$.  This gives a fully explicit nine-qubit instance produced by the NP-hardness reduction.

\section{Noncommuting templates: NP-complete at \texorpdfstring{$k=3$}{k=3}}

We now remove the commutativity assumption and show that decision CTC is NP-complete already for a template of fixed size $k=3$.  Our reduction uses the following permutation-group knapsack problem.  Given permutations $\pi_1,\pi_2,\pi_3,\pi\in S_m$, decide whether there exist nonnegative integers $e_1,e_2,e_3$ such that
\begin{equation}
  \pi_3^{e_3}\pi_2^{e_2}\pi_1^{e_1}=\pi.
  \label{eq:permutation-knapsack}
\end{equation}
Membership in a product of three cyclic permutation groups is NP-complete \cite{lohrey2025}.  Reversing the list of input generators gives the displayed chronological convention without changing the complexity.  Because each permutation has finite order, every exponent in Eq.~\eqref{eq:permutation-knapsack} has an equivalent canonical representative satisfying
\begin{equation}
  0\le e_i<\operatorname{ord}(\pi_i).
\end{equation}
The orders and their binary encodings are computable in polynomial time from the cycle decompositions.

For a permutation $\sigma\in S_m$, let $W(\sigma)$ be the Clifford operation that applies the wire permutation $\sigma$ independently to each of $m$ blocks of $m$ qubits.  Thus the construction uses $m^2$ qubits indexed by $(j,\ell)\in[m]\times[m]$.  Each $W(\sigma)$ is a SWAP-only Clifford circuit of polynomial size, and the map $\sigma\mapsto W(\sigma)$ is a faithful representation.  Define X-type Paulis
\begin{equation}
  P_a=\prod_{j=1}^{m}X_{(j,j)},\qquad
  P_b=\prod_{j=1}^{m}X_{(j,\pi(j))}.
\end{equation}
For every $\sigma\in S_m$,
\begin{equation}
  W(\sigma)P_aW(\sigma)^\dagger
  =\prod_{j=1}^{m}X_{(j,\sigma(j))}.
  \label{eq:faithful-pauli-encoding}
\end{equation}
The $j$th block records the $j$th image of the permutation, so equality with $P_b$ holds if and only if $\sigma=\pi$.  Equation~\eqref{eq:faithful-pauli-encoding} embeds equality of the complete permutation into the image of a single Pauli operator.

\begin{theorem}
Decision CTC is NP-complete for the fixed template size $k=3$.  NP-hardness holds even when the three template operations are wire-permutation Clifford operations, hence SWAP-only circuits, and the input and target Paulis are X-type.
\end{theorem}

\begin{proof}
Membership in NP follows by periodicity and binary exponentiation.  Given an instance of Eq.~\eqref{eq:permutation-knapsack}, set $U_i=W(\pi_i)$ and choose $P_a,P_b$ as above.  The instantiated template maps $P_a$ to $P_b$ exactly when
\begin{equation}
  \pi_3^{e_3}\pi_2^{e_2}\pi_1^{e_1}=\pi.
\end{equation}
The construction has polynomial size, so NP-hardness follows from permutation-group 3-knapsack.
\end{proof}

\section{Commuting CTC: \texorpdfstring{$\mathrm{NP}\cap\mathrm{BQP}$}{NP cap BQP}}

\subsection{Quantum algorithm for unrestricted integer exponents}

We now keep the commuting Clifford--Pauli reachability equation of Sec.~III unchanged while enlarging the exponent domain from $\{0,1\}^k$ to $\Z_{\geq0}^k$, equivalently $\Z^k$ by periodicity.  Counterintuitively, this enlargement restores the Abelian-group closure absent from the Boolean cube.  We distinguish decision from search explicitly: decision CTC asks only whether a repetition vector exists, whereas search CTC asks for such a vector.  In the commuting setting both tasks reduce to finite-Abelian-group membership, but the constructive task additionally requires coordinates of the target in terms of the supplied generators.  In this section, $k$ is allowed to scale with the input size; the fixed-$k$ setting is included as a special case.

The binary self-reduction described after Theorem~1 does not extend in an evident way to unrestricted integer exponents, which are defined modulo potentially large generator orders.  For example, if an orbit has odd order, replacing $F$ by $F^2$ leaves the same cyclic subgroup and therefore does not distinguish the parity of an exponent.  More sharply, in the one-generator DLP embedding of Sec.~VI, decision can be yes for every nonzero target when the input element generates the full multiplicative group, while search still requires recovery of the discrete logarithm.  This illustrates why the binary self-reduction has no direct analogue here.  In particular, the BQP upper bound for commuting integer-exponent decision CTC does not by itself yield witness recovery; the Abelian hidden-subgroup procedure below directly returns the required exponent coordinates.

\begin{theorem}
For arbitrary template size $k$, commuting decision CTC is in $\mathrm{NP}\cap\mathrm{BQP}$.  Commuting search CTC is solvable in quantum polynomial time.  More strongly, when a solution exists, the algorithm produces a particular exponent solution and a basis of the full relation lattice.
\end{theorem}

\begin{proof}
Let $\boldsymbol F=(F_1,\ldots,F_k)$ and, for $\boldsymbol e=(e_1,\ldots,e_k)\in\Z^k$, define the multi-index notation
\begin{equation}
  \boldsymbol F^{\boldsymbol e}
  :=F_k^{e_k}\cdots F_1^{e_1}.
\end{equation}
Let
\begin{equation}
  W=\mathrm{span}_{\F_2}\{\boldsymbol F^{\boldsymbol e}a:\boldsymbol e\in\Z^k\}.
\end{equation}
The integer exponents describe the full periodic group orbit.  Since each $F_i$ has finite order, the same orbit is generated by nonnegative exponents.  Starting with $a$, repeatedly apply every $F_i$ to the current basis vectors and add any linearly independent result.  At most $2n$ vectors can be added, so this procedure terminates after polynomially many matrix-vector operations and produces a basis $w_1,\ldots,w_d$ of $W$.  During the procedure we record an exponent vector $\boldsymbol\alpha_j\in\Z^k$ such that
\begin{equation}
  w_j=\boldsymbol F^{\boldsymbol\alpha_j}a,
  \qquad
  \boldsymbol\alpha_j=(\alpha_{j,1},\ldots,\alpha_{j,k})\in\Z^k,
\end{equation}
with $w_1=a$ and $\boldsymbol\alpha_1=0$.

If $b\notin W$, no orbit element can equal $b$, and the answer is no.  Suppose therefore that $b\in W$.  Membership of $b$ in $W$ can be tested efficiently, but $b\in W$ does not by itself imply that there exists an $\boldsymbol e$ satisfying $b=\boldsymbol F^{\boldsymbol e}a$.

At this stage, $b=\boldsymbol F^{\boldsymbol e}a$ is an orbit-membership condition rather than a group-membership condition: the target $b$ is a vector, not a group element, and its value alone does not determine the action of $\boldsymbol F^{\boldsymbol e}$ on the full ambient space.  On the orbit span $W$, however, commutativity implies that the image of $a$ determines the action on every orbit vector, since
\begin{equation}
  \boldsymbol F^{\boldsymbol e}
  \boldsymbol F^{\boldsymbol\alpha_j}a
  =\boldsymbol F^{\boldsymbol\alpha_j}b.
\end{equation}
This observation motivates converting the target vector $b$ into an operator on $W$.  Define a linear map $C_b:W\to W$ on the chosen basis by
\begin{equation}
  C_b w_j = \boldsymbol F^{\boldsymbol\alpha_j} b.
\end{equation}
The right-hand side lies in $W$ because $b\in W$ and $W$ is invariant under every $F_i$.  Thus the matrix of $C_b$ in the basis $w_1,\ldots,w_d$ is computable by classical linear algebra.  Likewise, let $B_i=F_i|_W$ and compute its matrix in this basis.  Each $B_i$ is invertible, and the $B_i$ commute pairwise.  Write $\boldsymbol B=(B_1,\ldots,B_k)$ and $\boldsymbol B^{\boldsymbol e}:=B_k^{e_k}\cdots B_1^{e_1}$.

We claim that
\begin{equation}
  b=\boldsymbol F^{\boldsymbol e}a \quad\Longleftrightarrow\quad
  C_b=\boldsymbol B^{\boldsymbol e}.
\end{equation}
Indeed, if $b=\boldsymbol F^{\boldsymbol e}a$, then commutativity gives, for every basis vector $w_j$,
\begin{equation}
  C_bw_j=\boldsymbol F^{\boldsymbol\alpha_j}b
  =\boldsymbol F^{\boldsymbol\alpha_j}\boldsymbol F^{\boldsymbol e}a
  =\boldsymbol F^{\boldsymbol e}w_j.
\end{equation}
Hence $C_b=\boldsymbol B^{\boldsymbol e}$.  Conversely, if $C_b=\boldsymbol B^{\boldsymbol e}$, applying both sides to $w_1=a$ gives $b=C_ba=\boldsymbol B^{\boldsymbol e}a=\boldsymbol F^{\boldsymbol e}a$.  The CTC instance is therefore equivalent to constructive membership of $C_b$ in the finite Abelian matrix group
\begin{equation}
  G_B=\langle B_1,\ldots,B_k\rangle.
\end{equation}
If $C_b$ is singular or fails to commute with some $B_i$, it cannot belong to $G_B$ and we reject immediately.

It remains to solve this constructive membership problem.  Efficient quantum algorithms for constructive membership in finite Abelian groups are known through reductions to the Abelian hidden-subgroup problem \cite{ivanyos2003,childs2014}.  Matrix multiplication, inversion, equality testing, and binary powering in $\GL(d,\F_2)$ are classical polynomial-time operations, and matrices provide unique encodings of group elements.  We implement this reduction by applying the Cheung--Mosca finite-Abelian-group decomposition algorithm, whose quantum subroutines are Shor order finding and Abelian hidden-subgroup relation finding \cite{shor1997,cheung2001}, to
\begin{equation}
  G=\langle B_1,\ldots,B_k,C_b\rangle.
\end{equation}
It returns independent generators $g_1,\ldots,g_s$ and their orders $q_1,\ldots,q_s$, giving
\begin{equation}
  G\cong \Z_{q_1}\times\cdots\times\Z_{q_s}
\end{equation}
and the associated Abelian hidden-subgroup relation-finding routine expresses each input generator in these coordinates.  More precisely, it returns vectors $\beta_i\in\prod_{\ell=1}^s\Z_{q_\ell}$ and $c\in\prod_{\ell=1}^s\Z_{q_\ell}$ such that
\begin{equation}
  B_i=g_1^{\beta_{1i}}\cdots g_s^{\beta_{si}},
  \qquad
  C_b=g_1^{c_1}\cdots g_s^{c_s}.
\end{equation}
Here $g_\ell$ is the $\ell$th independent generator of $G$ and has order $q_\ell$.
This coordinate computation is a generalized discrete-logarithm problem in a finite Abelian group and can be solved by standard quantum algorithms for such groups \cite{ivanyos2003,childs2014}.  Membership of $C_b$ in $G_B$ is now equivalent to the system of linear congruences
\begin{equation}
  \sum_{i=1}^k e_i\beta_i=c.
  \label{eq:abelian-coordinate-congruences}
\end{equation}

To solve Eq.~\eqref{eq:abelian-coordinate-congruences} explicitly, choose integer representatives of the coordinates and form
\begin{equation}
  M=(\beta_1\ \cdots\ \beta_k)\in\Z^{s\times k},
  \qquad
  D=\operatorname{diag}(q_1,\ldots,q_s).
\end{equation}
Then Eq.~\eqref{eq:abelian-coordinate-congruences} holds if and only if there is a vector $\boldsymbol y\in\Z^s$ such that
\begin{equation}
  M\boldsymbol e-D\boldsymbol y=c.
  \label{eq:integer-lifted-system}
\end{equation}
Apply the polynomial-time Kannan--Bachem Smith-normal-form algorithm, including its unimodular multiplier matrices, to the integer matrix $(M\ -D)$ \cite{kannan1979}.  The Smith form decides whether Eq.~\eqref{eq:integer-lifted-system} is solvable and, when it is, returns an integer solution $(\boldsymbol e^{(0)},\boldsymbol y^{(0)})$.  The same multiplier matrices give a basis of the integer kernel of $(M\ -D)$.  Projecting those kernel vectors onto their first $k$ coordinates gives a finite generating set for the relation lattice; polynomial-time Hermite-normal-form reduction of the projected generators yields a lattice basis for the relation lattice
\begin{equation}
  \Lambda=\{\boldsymbol z\in\Z^k:\boldsymbol B^{\boldsymbol z}=I\}.
  \label{eq:relation-lattice}
\end{equation}
Consequently, all integer solutions are exactly the affine lattice $\boldsymbol e^{(0)}+\Lambda$, and the nonnegative solutions of the original formulation are $(\boldsymbol e^{(0)}+\Lambda)\cap\Z_{\geq 0}^k$.  Reducing each coordinate of $\boldsymbol e^{(0)}$ modulo $\operatorname{ord}(B_i)$ gives an equivalent nonnegative repetition vector.  The particular solution and relation-lattice basis contain more information than the single witness required by $\mathrm{CTC}_{\mathrm{search}}$.

The only quantum steps are therefore order finding, Abelian-group decomposition, and coordinate relation finding; after these, Eqs.~\eqref{eq:integer-lifted-system} and \eqref{eq:relation-lattice} are solved by classical Smith-normal-form computation.  Since $d\le 2n$, the group orders have polynomial-length binary encodings, and all matrix-group operations and controlled binary powers require time polynomial in $n$, $k$, and the input length, the entire procedure runs in quantum polynomial time.  This proves the FBQP upper bound for commuting search CTC; ignoring the returned witness gives the BQP upper bound for decision CTC.  The NP inclusion follows from the polynomial-size periodic representatives and efficient verification discussed after Definition~2.
\end{proof}

The binary-exponent NP-completeness result does not contradict this theorem.  Although $\{0,1\}^k\subset\Z^k$, restricting the exponents to binary values does not merely reduce the search space; it changes the decision problem.  For unrestricted integer exponents, the solution set is either empty or an affine lattice $\boldsymbol e^{(0)}+\Lambda$, and the task is finite-Abelian-group membership.  The binary-exponent problem instead asks whether this affine lattice intersects the Boolean cube $\{0,1\}^k$.  In general, this Boolean constraint removes the group closure exploited by the quantum algorithm and permits the NP-hardness reduction; commuting self-inverse Clifford actions form an exception, analyzed in Sec.~\ref{sec:involutions}.  Once arbitrary integer exponents are allowed, that reduction no longer applies, whereas the restored Abelian group structure gives the BQP upper bound proved above.  Concretely, in the clause gadget of Sec.~III unrestricted exponents enforce only $s_c\equiv1\pmod 4$, because $G_c^4=I$, rather than the Boolean exactly-one condition $s_c=1$.  Commutativity is essential to the upper bound because $\boldsymbol e\mapsto\boldsymbol F^{\boldsymbol e}$ is not a homomorphism for noncommuting generators.  Indeed, the preceding section proves NP-completeness at $k=3$.  The case $k=2$ is not settled by that reduction: two-generator knapsack for permutation groups is classically polynomial-time solvable \cite{lohrey2025}, while no general BQP algorithm is known for the broader two-operation Clifford--Pauli orbit problem.  We therefore leave noncommuting CTC at $k=2$ open.

\subsection{Self-inverse Clifford actions: reachability versus repetition cost}
\label{sec:involutions}

An instructive boundary case occurs when every phase-free Clifford action is self-inverse, $F_i^2=I$, in addition to pairwise commutativity.  This includes exactly commuting Clifford operations satisfying $U_i^2=I$.  Here reducing each exponent modulo two preserves the action, so binary and unrestricted integer exponents give the same reachability problem.  Nevertheless, imposing an upper bound on the total number of template-operation applications changes its complexity.

\begin{theorem}[Classical reachability for commuting self-inverse Clifford actions]
\label{thm:involution-reachability}
For pairwise commuting $F_1,\ldots,F_k\in\Sp(2n,\F_2)$ satisfying $F_i^2=I$, both decision CTC and recovery of one binary exponent solution are solvable in deterministic classical polynomial time, with $k$ part of the input.  More strongly, the algorithm returns the full binary solution set as an affine subspace of $\F_2^k$.
\end{theorem}

\begin{proof}
We construct the solution set by successively reducing the allowed error $\boldsymbol F^{\boldsymbol e}a-b$.  Although this error is generally nonlinear in the exponents, each refinement will require only a linear system over $\F_2$.

Write $V=\F_2^{2n}$ and introduce $N_i=F_i-I$ to express each binary choice as $F_i^{e_i}=I+e_iN_i$.  Since the field has characteristic two, $N_i^2=F_i^2+I=0$, and the $N_i$ commute.  Products of the factors $I+e_iN_i$ nevertheless contain higher-degree terms; for example,
\begin{align*}
  F_1^{e_1}F_2^{e_2}a
  &=a+e_1N_1a+e_2N_2a\\
  &\quad+e_1e_2N_1N_2a.
\end{align*}
To organize these terms, define the descending sequence of subspaces
\begin{equation}
  V_0=V,\qquad V_{t+1}=\sum_{i=1}^k N_iV_t.
  \label{eq:involution-filtration}
\end{equation}
The space $V_t$ is spanned by images of products of $t$ of the $N_i$; in particular, a term containing $t$ such factors lies in $V_t$.  Every product of $k+1$ factors vanishes, because some factor repeats and the factors commute.  Thus $V_{k+1}=0$.  These spaces are invariant under all generators, and their bases can be computed recursively by Gaussian elimination without enumerating products of the $N_i$.

The key property is that, for $x\in V_t$,
\[
  F_ix-x=N_ix\in V_{t+1}.
\]
Hence every $F_i$, and every $\boldsymbol F^{\boldsymbol e}$, acts as the identity on $V_t/V_{t+1}$: applying it to a vector in $V_t$ changes that vector only by an element of $V_{t+1}$.  This property will make the correction at each stage linear.

We use $V_t$ as the space of allowed errors and compute
\begin{equation}
  S_t=\{\boldsymbol e\in\F_2^k:
    \boldsymbol F^{\boldsymbol e}a-b\in V_t\}.
\end{equation}
The unknown exponents always belong to $\F_2^k$; $V_t$ specifies how closely their action must match $b$.  Initially $S_0=\F_2^k$, while $S_{k+1}$ is the exact solution set because $V_{k+1}=0$.

To describe each nonempty $S_t$, define
\begin{equation}
  K_t=\{\boldsymbol v\in\F_2^k:
    (\boldsymbol F^{\boldsymbol v}-I)a\in V_t\}.
\end{equation}
Commutativity and $F_i^2=I$ give $\boldsymbol F^{\boldsymbol v+\boldsymbol w}=\boldsymbol F^{\boldsymbol v}\boldsymbol F^{\boldsymbol w}$ for binary exponent vectors.  Thus $K_t$ is the stabilizer of $a+V_t$ under the additive group $\F_2^k$, and is a linear subspace.  If $\boldsymbol e^{(t)}\in S_t$, invariance of $V_t$ gives $S_t=\boldsymbol e^{(t)}+K_t$.  The algorithm maintains this representation, starting with $\boldsymbol e^{(0)}=0$ and $K_0=\F_2^k$.

Suppose this representation of $S_t$ is known.  Let $\boldsymbol v_1,\ldots,\boldsymbol v_q$ be a basis of $K_t$.  Every candidate that preserves the current error condition has the form $\boldsymbol e^{(t)}+\sum_jc_j\boldsymbol v_j$, with $c_j\in\F_2$.  Define its elementary correction vectors by
\begin{equation}
  \delta_j=(\boldsymbol F^{\boldsymbol v_j}-I)a\in V_t.
\end{equation}
The map $\boldsymbol v\mapsto(\boldsymbol F^{\boldsymbol v}-I)a$ becomes linear when restricted to $K_t$ and reduced modulo $V_{t+1}$.  Indeed, for $\boldsymbol v,\boldsymbol w\in K_t$,
\begin{align}
  (\boldsymbol F^{\boldsymbol v+\boldsymbol w}-I)a
  &=(\boldsymbol F^{\boldsymbol v}-I)a
    +\boldsymbol F^{\boldsymbol v}(\boldsymbol F^{\boldsymbol w}-I)a
    \nonumber\\
  &\equiv(\boldsymbol F^{\boldsymbol v}-I)a
    +(\boldsymbol F^{\boldsymbol w}-I)a
    \pmod{V_{t+1}}.
\end{align}
The congruence uses $(\boldsymbol F^{\boldsymbol w}-I)a\in V_t$ and the identity action of $\boldsymbol F^{\boldsymbol v}$ on $V_t/V_{t+1}$.

Let $r_t=b-\boldsymbol F^{\boldsymbol e^{(t)}}a\in V_t$ be the error to correct.  For $\boldsymbol v=\sum_jc_j\boldsymbol v_j$, the new residual is
\begin{align*}
  \boldsymbol F^{\boldsymbol e^{(t)}+\boldsymbol v}a-b
  &=-r_t+\boldsymbol F^{\boldsymbol e^{(t)}}(\boldsymbol F^{\boldsymbol v}-I)a\\
  &\equiv-r_t+\sum_jc_j\delta_j\pmod{V_{t+1}}.
\end{align*}
Here $\boldsymbol F^{\boldsymbol e^{(t)}}$ also acts as the identity on $V_t/V_{t+1}$.  Consequently, the corrected exponent vector belongs to $S_{t+1}$ exactly when
\begin{equation}
  \sum_{j=1}^q c_j\delta_j
  \equiv b-\boldsymbol F^{\boldsymbol e^{(t)}}a
  \pmod{V_{t+1}}.
  \label{eq:involution-linear-step}
\end{equation}
Choosing coordinates on $V_t/V_{t+1}$ turns Eq.~\eqref{eq:involution-linear-step} into a linear system over $\F_2$ for the correction coefficients $c_j$.  If it is inconsistent, then $S_{t+1}$ is empty and no exact solution exists.  Otherwise, a particular solution $\boldsymbol c^{(0)}$ gives
\[
  \boldsymbol e^{(t+1)}
  =\boldsymbol e^{(t)}+\sum_jc_j^{(0)}\boldsymbol v_j.
\]
Map a basis of the homogeneous solution space through $\boldsymbol c\mapsto\sum_jc_j\boldsymbol v_j$ to obtain a basis of $K_{t+1}$.  This yields the complete affine set $S_{t+1}=\boldsymbol e^{(t+1)}+K_{t+1}$, retaining all solutions at the finer error tolerance.

After at most $k+1$ refinements, the allowed error space is zero and the resulting affine set is exactly the set of binary CTC solutions.  Each linear system has at most $k$ unknowns and $2n$ equations, and all required matrix products and subspace computations have polynomial size.  The algorithm is therefore deterministic and classically polynomial-time.  It returns one representative and a basis of the homogeneous solution space, rather than enumerating the possibly exponentially many solutions.  For $a=0$, a solution exists exactly when $b=0$, which can also be checked directly.
\end{proof}

To specify the cost constraint, let $m\in\Z_{\geq0}$ be an additional input, encoded in binary.  We ask whether there is a solution of Eq.~\eqref{eq:ctc} with
\begin{equation}
  \boldsymbol e\in\Z_{\geq0}^k,
  \qquad \sum_{i=1}^k e_i\leq m.
  \label{eq:repetition-budget}
\end{equation}
One application of a supplied operation $U_i$ counts as one unit of cost; this counts template-operation uses, rather than the elementary gates inside each $U_i$.

\begin{theorem}[NP-completeness with a repetition budget]
\label{thm:involution-budget}
For templates with commuting self-inverse Clifford actions, deciding whether Eq.~\eqref{eq:repetition-budget} can be satisfied is NP-complete when $k$ is part of the input.  This holds for either binary or nonnegative integer exponents, even for exactly commuting CNOT-only operations satisfying $U_i^2=I$ and X-type input and target Paulis.  Consequently, minimizing the total number of template-operation applications is NP-hard.
\end{theorem}

\begin{proof}
We reduce from syndrome decoding: given $H\in\F_2^{r\times k}$, $s\in\F_2^r$, and $m$, decide whether there exists a binary vector $\boldsymbol e$ with
\begin{equation}
  H\boldsymbol e=s,\qquad
  \operatorname{wt}(\boldsymbol e)\leq m.
  \label{eq:syndrome-decoding}
\end{equation}
This decision problem is NP-complete \cite{berlekamp1978}.  On $r+1$ qubits labeled $0,1,\ldots,r$, define
\begin{equation}
  \begin{aligned}
    U_i&=\prod_{\{j:H_{ji}=1\}}\mathrm{CNOT}_{0\rightarrow j},\\
    P_a&=X_0,\qquad P_b=X_0\prod_{j=1}^r X_j^{s_j}.
  \end{aligned}
  \label{eq:syndrome-template}
\end{equation}
CNOT gates with a common control commute exactly; consequently the $U_i$ commute exactly and satisfy $U_i^2=I$.  Conjugating $X_0$ by $U_i$ appends precisely the X support specified by the $i$th column of $H$.  Therefore, for arbitrary nonnegative integer exponents,
\begin{equation}
  U_{\mathcal T}(\boldsymbol e)P_a
  U_{\mathcal T}(\boldsymbol e)^\dagger=P_b
  \quad\Longleftrightarrow\quad
  H(\boldsymbol e\bmod2)=s.
\end{equation}
For binary exponents, the repetition budget is exactly the weight bound in Eq.~\eqref{eq:syndrome-decoding}.  For nonnegative integer exponents, replacing each $e_i$ by $e_i\bmod2$ preserves the Pauli action and never increases the cost.  Hence the same equivalence holds with the budget in either exponent domain.  The construction uses at most $rk$ CNOT gates and has polynomial size, proving NP-hardness.  Membership in NP follows because every yes-instance has a binary witness of length $k$, whose action and cost can be checked in polynomial time.
\end{proof}

In Eqs.~\eqref{eq:syndrome-template}, feasibility and recovery of one binary solution require only solving $H\boldsymbol e=s$ by Gaussian elimination.  The full integer solution set, when nonempty, is
\begin{equation}
  \boldsymbol e^{(0)}+\Lambda_H,\qquad
  \Lambda_H=\{\boldsymbol z\in\Z^k:H\boldsymbol z\equiv0\pmod2\}.
\end{equation}
A basis of $\Lambda_H$ is also computable in classical polynomial time by integer normal forms \cite{kannan1979}.  Nevertheless, selecting a solution with at most $m$ operation uses is NP-complete.  Thus this result separates finding a feasible instantiation from finding a low-cost one, even when neither feasibility nor the relation lattice requires a quantum computation.  It is distinct from the binary-feasibility hardness in Sec.~III: for self-inverse actions, the binary restriction alone preserves the full generated group, and the additional weight bound causes the hardness.

\section{One-operation templates and discrete logarithm}

\subsection{Qubit case: discrete logarithm over \texorpdfstring{$\F_{2^r}$}{F2r}}

The $k=1$ specialization of the same unrestricted search relation, $F^e a=b$, already captures a standard Shor-type problem.  It also makes the decision--search distinction concrete: subgroup membership and recovery of an exponent are different computational tasks.  First consider the qubit case.  Let $\F_{2^r}$ be represented as an $r$-dimensional vector space over $\F_2$.  For $g\in \F_{2^r}^{\times}$, multiplication by $g$ is an invertible binary linear map
\begin{equation}
  A_g\in \GL(r,\F_2).
\end{equation}
Embed this into a Clifford symplectic matrix by
\begin{equation}
  F_g=
  \begin{pmatrix}
  A_g&0\\
  0&A_g^{-T}
  \end{pmatrix}
  \in \Sp(2r,\F_2).
\end{equation}
Quantum circuits for DLP over $\F_{2^r}$, including explicit controlled-multiplication constructions, were studied by Beauregard, Brassard, and Fernandez \cite{beauregard2003}.  Here the corresponding uncontrolled fixed-multiplier map has a particularly simple Clifford description.  The matrix $A_g$ and its Clifford circuit can be constructed explicitly.  Fix an $\F_2$ basis $\beta=(\beta_1,\ldots,\beta_r)$ of $\F_{2^r}$.  The $j$th column of $A_g$ is the coordinate vector $[g\beta_j]_\beta$, so all entries of $A_g$ are obtained by arithmetic in the chosen representation of the field.  Multiplication by a fixed nonzero field element is $\F_2$-linear, so its binary matrix can be synthesized as a CNOT-and-SWAP circuit by standard linear reversible-circuit synthesis \cite{patel2008}.  Concretely, Gaussian elimination decomposes $A_g$ into elementary binary row operations.  The row addition $x_i\leftarrow x_i+x_j$ is implemented by $\operatorname{CNOT}_{j\to i}$, while exchanging rows $i$ and $j$ is implemented by $\operatorname{SWAP}_{i,j}$.  Reversing the elimination sequence therefore gives a CNOT-and-SWAP circuit for $A_g$; each SWAP may in turn be replaced by three CNOTs.  If this reversible circuit implements $\lvert x\rangle\mapsto\lvert A_gx\rangle$, its Pauli action is automatically $\operatorname{diag}(A_g,A_g^{-T})$, so the Z-type block need not be synthesized separately.  The same procedure applies to every power $A_g^{2^j}$ used in phase estimation: the matrix power is first computed classically by repeated squaring and is then synthesized directly, without repeating the circuit $2^j$ times.

For X-type Paulis, only the $A_g$ block acts.  Taking $a$ to represent $1\in\F_{2^r}$ and $b$ to represent $h\in\F_{2^r}^{\times}$ gives
\begin{equation}
  F_g^e a=b
  \quad\Longleftrightarrow\quad
  g^e=h.
\end{equation}
Thus, for this binary-linear family, the resulting one-operation search-CTC instance is exactly a finite-field discrete logarithm instance.

\begin{corollary}
Search CTC for a one-operation Clifford template over qubits is discrete-logarithm hard under polynomial-time reductions.  The corresponding decision problem asks only whether $h$ belongs to the cyclic subgroup generated by $g$, whereas search recovers an exponent.  Thus the DLP hardness applies to $\mathrm{CTC}_{\mathrm{search}}$, not to $\mathrm{CTC}_{\mathrm{dec}}$.
\end{corollary}

Together with Theorem~1, this places a DLP-containing integer-search slice and an NP-complete binary-decision slice in adjacent parameter regimes of the same Clifford--Pauli reachability relation.

As a concrete toy example, take
\begin{equation}
  \F_8=\F_2[\alpha]/(\alpha^3+\alpha+1).
\end{equation}
Consider the discrete-logarithm problem
\begin{equation}
  \alpha^e=1+\alpha+\alpha^2.
  \label{eq:f8-dlp}
\end{equation}
In the basis $(1,\alpha,\alpha^2)$, multiplication by $\alpha$ is
\begin{equation}
  A=
  \begin{pmatrix}
  0&0&1\\
  1&0&1\\
  0&1&0
  \end{pmatrix},
  \qquad A^7=I.
\end{equation}
For this example, a direct synthesis of $A$ is, in chronological order, $\operatorname{SWAP}_{1,2}$, $\operatorname{SWAP}_{1,3}$, and $\operatorname{CNOT}_{1\to 2}$.  Indeed, this sequence maps $(x_1,x_2,x_3)$ to $(x_3,x_1+x_3,x_2)=A(x_1,x_2,x_3)$.  It maps X supports as
\begin{align}
  X_1&\mapsto X_2\mapsto X_3\mapsto X_1X_2
  \mapsto X_2X_3 \nonumber\\
  &\mapsto X_1X_2X_3
  \mapsto X_1X_3\mapsto X_1.
\end{align}
Explicitly, define the three-qubit Clifford operation
\begin{equation}
  U=\operatorname{CNOT}_{1\to 2}
    \operatorname{SWAP}_{1,3}\operatorname{SWAP}_{1,2}.
  \label{eq:f8-template-operation}
\end{equation}
Then Eq.~\eqref{eq:f8-dlp} is the search-CTC instance with the one-operation template $\mathcal T=(U)$, input Pauli $P=X_1$, and target Pauli $P'=X_1X_2X_3$: find $e\in\Z_{\geq0}$ satisfying
\begin{equation}
  U^eX_1(U^\dagger)^e=X_1X_2X_3.
  \label{eq:f8-ctc}
\end{equation}
This condition is equivalent to $\alpha^e=1+\alpha+\alpha^2$, and holds exactly when $e\equiv5\pmod7$.  Thus the canonical solution is $e=5$, corresponding to five applications of the supplied operation $U$.

\subsection{Extension to odd prime dimensions}

Although the paper's main setting is qubits, the one-operation-template DLP embedding is not essentially binary.  For a prime-dimensional qudit Clifford system, Pauli labels live in $\F_p^{2n}$ and Clifford conjugation acts through $\Sp(2n,\F_p)$.  Let $p$ be prime and let $\F_{p^r}$ be represented as an $r$-dimensional vector space over $\F_p$.  Multiplication by $g\in\F_{p^r}^{\times}$ is an $\F_p$-linear map
\begin{equation}
  A_g\in\GL(r,\F_p),
\end{equation}
which defines a prime-dimensional qudit Clifford action
\begin{equation}
  F_g=
  \begin{pmatrix}
  A_g&0\\
  0&A_g^{-T}
  \end{pmatrix}
  \in\Sp(2r,\F_p).
\end{equation}
For X-type Paulis,
\begin{equation}
  F_g^e a=b
  \quad\Longleftrightarrow\quad
  g^e=h
\end{equation}
after identifying $a$ with $1$ and $b$ with $h\in\F_{p^r}^{\times}$.  In particular, setting $r=1$ gives the prime-field discrete logarithm problem in $\F_p^\times$ as a one-qudit Clifford-orbit problem.

For fixed small characteristic, including $\F_{2^r}$, the best known classical algorithms are quasi-polynomial in important asymptotic regimes \cite{barbulescu2014}.  Prime-field DLP over $\F_p^\times$ is instead attacked by number-field-sieve-type algorithms with subexponential complexity comparable in shape to integer factoring, and no classical polynomial-time algorithm is known.

\onecolumngrid
\vspace{6pt}
\begin{center}
\begin{minipage}{0.92\textwidth}
\refstepcounter{table}
\label{tab:complexity-summary}
\small
\noindent\textsc{Table \thetable.} Summary of the problem settings and complexity results.  The exponent domain distinguishes the binary-exponent variant from unrestricted CTC.  All rows use the reachability relation in Eq.~\eqref{eq:ctc}; the listed regimes vary the exponent domain, commutativity, template size, generator orders, repetition budget, and decision versus search output.  The budget $m$ counts applications of the supplied template operations.
\par\smallskip
\centering
\begin{tabular}{@{}p{0.18\linewidth}p{0.14\linewidth}c p{0.13\linewidth}p{0.37\linewidth}@{}}
\hline\hline
\raggedright Template class & \raggedright Commutativity & $k$ & \raggedright Exponent domain & \raggedright Complexity or status \tabularnewline
\hline
\raggedright Decision qubit CTC
  & \raggedright Commuting & Variable
  & $\{0,1\}$
  & \raggedright NP-complete, even for CNOT-only operations \tabularnewline
\raggedright Search qubit CTC
  & \raggedright Commuting & Variable
  & $\{0,1\}$
  & \raggedright In FNP; reducible to decision with at most $k$ queries \tabularnewline
\raggedright Self-inverse CTC, $F_i^2=I$
  & \raggedright Commuting & Variable
  & $\{0,1\}$ or $\Z_{\geq0}$
  & \raggedright Classical polynomial-time decision and recovery of one binary solution \tabularnewline
\raggedright Self-inverse CTC, $\sum_i e_i\leq m$
  & \raggedright Commuting & Variable
  & $\{0,1\}$ or $\Z_{\geq0}$
  & \raggedright NP-complete, even for exactly commuting CNOT-only operations with $U_i^2=I$ \tabularnewline
\raggedright Decision qubit CTC
  & \raggedright Commuting & Arbitrary
  & $\Z$
  & \raggedright In $\mathrm{NP}\cap\mathrm{BQP}$ \tabularnewline
\raggedright Search qubit CTC
  & \raggedright Commuting & Arbitrary
  & $\Z$
  & \raggedright Quantum polynomial time; a particular solution and the full relation lattice returned; DLP-hard already at $k=1$ \tabularnewline
\raggedright Decision qubit CTC
  & \raggedright Noncommuting & $3$
  & $\Z$
  & \raggedright NP-complete, even for SWAP-only operations \tabularnewline
\raggedright Decision qubit CTC
  & \raggedright Noncommuting & $2$
  & $\Z$
  & \raggedright Open \tabularnewline
\raggedright Search prime-dimensional qudit CTC
  & \raggedright Automatic & $1$
  & $\Z$
  & \raggedright Contains DLP over $\F_{p^r}^{\times}$, including $\F_p^{\times}$; solvable in quantum polynomial time \tabularnewline
\hline\hline
\end{tabular}
\end{minipage}
\end{center}
\vspace{6pt}
\twocolumngrid

\section{Discussion}

The sharpest transition occurs without changing the phase-free commuting Clifford--Pauli reachability equation.  Allowing arbitrary integer repetitions restores Abelian-group closure: decision CTC lies in $\mathrm{NP}\cap\mathrm{BQP}$, and a quantum algorithm constructively recovers one solution and the full exponent-relation lattice.  Its one-operation search slice already contains binary finite-field DLP.  These results give an affirmative answer to the question posed in the subtitle for commuting templates with unrestricted integer exponents: quantum computers can recover classically verifiable compilation witnesses in polynomial time, including for instances encoding finite-field DLP.  This demonstrates algorithmic usefulness in this setting, although it does not establish an unconditional quantum speedup over classical computation.  In general, imposing $\boldsymbol e\in\{0,1\}^k$ removes that closure and asks whether an affine lattice intersects the Boolean cube, a decision problem that is NP-complete for variable $k$ even for exactly commuting CNOT-only templates.  The DLP and NP-completeness embeddings use different values of $k$ and different output tasks, but they occupy adjacent regimes of the same reachability relation.  Both are realized within CNOT-only phase-free Clifford dynamics.  This provides an unusually unified comparison between Shor-type structured witness recovery and NP-complete Boolean selection, rather than juxtaposing unrelated problem encodings.  As a second axis, dropping commutativity makes integer-exponent decision CTC NP-complete already for three wire-permutation Cliffords.

The self-inverse restriction in Sec.~\ref{sec:involutions} isolates a separate feasibility--cost transition.  For commuting $F_i$ with $F_i^2=I$, allowing arbitrary integer exponents does not enlarge the binary orbit, and one feasible exponent vector can be found classically in polynomial time.  Even in this tractable regime, the requirement $\sum_i e_i\leq m$ makes decision NP-complete.  The syndrome-decoding construction shows that the difficulty can persist after a particular solution and the full relation lattice are already known.  Thus an efficient feasibility algorithm need not yield an implementation with a small number of operation uses.

From a resource perspective, the binary-linear DLP embedding may also be a convenient early fault-tolerant benchmark.  Unlike modular multiplication in Shor factoring, whose carries and modular reduction require substantial non-Clifford resources \cite{gidney2019}, each uncontrolled map $x\mapsto A_gx$ is a CNOT-and-SWAP circuit \cite{patel2008}; non-Clifford gates enter when these maps are coherently controlled.  This observation alone does not establish quantum advantage, especially because quasi-polynomial classical algorithms are known for binary-field DLP in important asymptotic regimes \cite{barbulescu2014}, but it motivates a more detailed future resource comparison.

Several directions remain open.  First, the most immediate complexity question is the noncommuting two-operation case, which lies between the DLP-containing one-operation problem and the NP-complete three-operation problem.  Second, the present paper uses Pauli reachability as a minimal template-compilation subproblem; extending the classification to full logical Clifford transformations would connect more directly to code-block compilation.  Third, the BQP algorithm should be made fully explicit at the circuit level for families of binary linear maps with low CNOT volume and low Toffoli depth.  Finally, one should identify families whose best known classical algorithms remain costly while preserving the low-non-Clifford controlled-linear structure.  Such families could provide a practical and verifiable benchmark for early fault-tolerant quantum processors.

\begin{samepage}
\begin{acknowledgments}
KF is supported by MEXT Quantum Leap Flagship Program (MEXT Q-LEAP) Grant No. \mbox{JPMXS0120319794}, JST COI-NEXT Grant No. \mbox{JPMJPF2014}, JST Moonshot R\&D Grant No. \mbox{JPMJMS256E}, and JST CREST \mbox{JPMJCR24I3}.
\end{acknowledgments}
\end{samepage}

\bibliography{references}

\end{document}